\documentclass[11pt, a4paper]{article}

\usepackage{latexsym}
\usepackage{amsmath}
\usepackage{epsf}
\usepackage[latin1]{inputenc}
\usepackage{enumerate}
\usepackage{calc}
\usepackage{amssymb, graphicx}
\usepackage[active]{srcltx}

\newenvironment{proof}[1][{}]{\par\noindent\bgroup{\bf Proof#1. }}{\egroup\hspace*{\fill}\fertig}

\renewcommand{\o}[1]{\overline{{#1}}}
\renewcommand{\t}[1]{\widetilde{{#1}}}

\renewcommand{\div}{div}

\newtheorem{theorem}{Theorem}[section]
\newtheorem{lemma}[theorem]{Lemma}
\newtheorem{proposition}[theorem]{Proposition}

\def\fertig{{\unskip\nobreak\hfil\penalty 50
   \hskip2em\hbox{}\nobreak\hfill$\square$\medskip}}

\newcommand{\p}{\varphi}
\newcommand{\la}{\lambda}
\renewcommand{\to}{\longrightarrow}
\renewcommand{\O}{\Omega}
\renewcommand{\a}{\alpha}

\newcommand{\D}{\nabla}

\def\diag{\text{diag }}

\def\R{\mathbb{R}}

\def\NN{\mathbb{N}}

\def\SS{\mathbb{S}}
\def\DD{\mathcal{D}}
\def\W{\mathcal{W}}

\def\D{\nabla}
\def\a{\alpha}

\DeclareMathOperator{\Tr}{Tr}

\DeclareMathOperator{\sym}{sym}
\DeclareMathOperator{\cof}{cof}

\begin{document}

\title{A remark on constrained von K\'arm\'an theories}

\author{Peter Hornung
\footnote{
Email: hornung@iam.uni-bonn.de
}
}

\date{}

\maketitle

\begin{abstract}
We derive the Euler-Lagrange equation corresponding
to a variant of `non-Euclidean' constrained von K\'arm\'an theories. 
\end{abstract}

\section{Introduction}

F\"oppl-von-K\'arm\'an theories arise as asymptotic theories modelling the
behaviour of thin elastic films, in an energy regime allowing only for very small deformations.
The elastic energy of such deformations (with respect to the thickness of the film) is therefore much lower than
that of generic nonlinear bending deformations. The asymptotic behaviour of the latter
is modelled by the fully nonlinear Kirchhoff plate theory.
We refer to \cite{fjm1, fjm2} for a derivation and thorough discussion of these theories, cf. also
\cite{Ciarlet2}.
\\
More precisely, the asymptotic behaviour of thin film 
deformations whose elastic energy lies in a regime just below
the nonlinear bending regime is captured by so-called constrained von K\'arm\'an theories,
cf. \cite{fjm2}. Their behaviour is essentially fully described by their out-of-plane displacement 
$v : S\to\R$, where $S\subset\R^2$ is the reference configuration of the sample.
The asymptotic elastic energy of such a displacement $v$ is then given by
\begin{equation}\label{null} 
\frac{1}{24}\int_S Q_2\left( \D^2 v(x) \right)\ dx + \int_S f\cdot v\ dx,
\end{equation}
subject to the constraint
\begin{equation}\label{MA} 
\det(\D^2 v) = 0.
\end{equation}
Here $Q_2$ is the quadratic form of linearised elasticity and $f$ models
applied forces.
\\
Motivated by applications in non-Euclidean (or pre-strained) elasticity (cf. e.g.
\cite{EfratiSharonKupferman} and \cite{LeMaPa}), 
we consider variants of functionals as in \eqref{null} by allowing a nonzero right-hand side 
in \eqref{MA}. For simplicity, we
restrict to the isotropic case when $Q_2 = |\cdot|^2$ and we do not consider forces. More general
situations can be handled in the same way, as our main focus is on the constraint 
\begin{equation}
\label{MA-1}
\det\D^2 v = k 
\end{equation}
itself. The problem is therefore to understand, 
on a bounded domain $S\subset\R^{2}$, and for given $k : S\to\R$,
the functional
$$
\W_k(v) = 
\begin{cases}
\int_S |\D^{2} v(x)|^{2}\ dx &\mbox{ if }v\in W_k^{2,2}(S)
\\
+ \infty &\mbox{ otherwise.}
\end{cases}
$$
Here
$$
W^{2,2}_k = \{v\in W^{2,2}(S) : \det\D^2 v = k \mbox{ pointwise almost everywhere } \}.
$$
We use a similar notation for other function spaces, such as $C^{2,\a}_k(\o S)$ or $C^2_k(\o S)$. 
The Monge-Amp\`ere equation $\det\D^{2} v = k$ 
has been studied extensively over the last decades. We refer to the book
\cite{GT} for a list of references on the topic.
\\
The functionals $\W_k$ are scalar variants of the functionals studied
in \cite{H-general, H-convex}. The purpose of this note is to show how the approach developed 
in those papers can be adapted to the simpler situation considered here. In passing, we provide
here a classical functional analytic framework for this sort of problems.
Our main focus is on the elliptic case ($k > 0$), which is the simplest one.
The methods are, therefore, very basic. Indeed, in this situation, soft arguments readily yield the desired Euler-Lagrange equation.
\\
At the end of the note we discuss the cases when $k$ is constant.

\section{Main results}

For simplicity, we assume throughout this note that $S\subset\R^2$ is a simply connected, bounded domain with a smooth boundary,
and we let $k\in C^{\infty}(\o S)$.

\subsection{Existence of minimisers}

As in \cite{H-general}, existence of minimisers can be proven by a robust and straightforward
argument.
\begin{proposition}
\label{existence} 
The functional $\W_k$
attains a minimum in the space $W_k^{2,2}(S)$.
\end{proposition}
\begin{proof}
We only need to consider the case when the infimum of $\W_k$ is finite.
But then the result is a straightforward application of the direct method on the space
$$
X = \{v\in W^{2, 2} : \int_S v = 0 \mbox{ and }\int_S \D v = 0 \}.
$$
In fact, $\W_k$ is obviously $W^{2,2}$-coercive and lower semicontinuous under weak $W^{2,2}$-convergence.
But the constraint is stable under weak $W^{2,2}$-convergence, because
the determinant is continuous under weak $W^{2,2}$-convergence. Applying Poincar\'e's inequality,
we obtain the existence of a minimiser in $X$.
\end{proof}

It is clear that the same proof also works for general domain dimensions, other energy densities,
additional force terms, boundary conditions, etc.
As in \cite{H-convex}, when $k > 0$ then one has a better existence result: 
\begin{proposition}
\label{ellex} 
Assume that $k > 0$ on $\o S$. Then the functional
$\W_k$ attains a minimum on the set
\begin{equation}
\label{ellex-1}
 W^{2,2}_k(S)\cap C^{\infty}(S).
\end{equation}
\end{proposition}
\begin{proof}
First note that functions $v$ belonging to the set \eqref{ellex-1}
are either uniformly convex or uniformly concave, and that the infimum
of $\W_k$ is the same on both of these components of the set \eqref{ellex-1}. So 
we will prove that the minimimum is attained on the set
$$
X = \{v\in W^{2,2}_k(S)\cap C^{\infty}(S) : v\mbox{ is convex }\}.
$$
In fact, by interior regularity for convex Monge-Amp\`ere equations we have
$$
X = \{v\in W_k^{2,2}(S) : v\mbox{ is convex }\}.
$$
Hence the space $X$ is closed under weak $W^{2, 2}$-convergence and therefore we
can find a minimiser in this space by the same arguments as in proof of Proposition \ref{existence}.
\end{proof}

\subsection{Lagrange multiplier rule for the elliptic case}

The formal Lagrange multiplier rule asserts that critical points of $\W_k$
are critical for the functional
\begin{equation}
\label{formal-1} 
v\mapsto \int_S |\D^2 v|^2 - \lambda\det\D^2 v
\end{equation} 
without additional constraints, cf. e.g. \cite{GuvenMuller} for a related situation. Here $\la$ is some Lagrange multiplier. The
Euler-Lagrange equation corresponding to \eqref{formal-1} is
$$
\div\div\left( \D^2 v - \la\cof\D^2 v \right) = 0,
$$
or, since $\div\cof\D = 0$,
$$
\Delta^2 v - \cof\D^2 v : \D^2\la = 0.
$$
We will show that, under suitable regularity assumptions, 
this formal Lagrange multiplier rule can be
justified by means of very soft functional analytic arguments.
\\
For a rigorous approach we introduce the following notions, which
are variants of those introduced in \cite{H-general}:
A function $v\in W^{2,2}_k(S)$ is said to be stationary for $\W_k$ if 
$$
\frac{d}{dt}\Big|_{t = 0} \int_S |\D^2 u(t)|^2 = 0
$$
for all (strongly $W^{2,2}$-continuous, say) 
maps $t\mapsto u(t)$ from a neighbourhood of zero in $\R$ into $W^{2,2}_k(S)$
such that $u(0) = v$ and such that the derivative $u'(0)$ exists.
\\
A function $v\in W^{2,2}_k(S)$ is said to be formally stationary for $\W_k$
if 
$$
\int_S \D^2 v : \D^2 h = 0 \mbox{ for all }h\in W^{2,2}(S)\mbox{ with }\cof\D^2 v : \D^2 h = 0\mbox{ a.e. in }S.
$$
Our main results for the elliptic case are the following two remarks.
\begin{proposition}\label{manif} 
Let $\a\in (0, 1)$ and let $k > 0$ on $\o S$. Then the set
$$
 C^{2,\a}_k(\o S) := \{u\in C^{2,\a}(\o S) : \det\D^2 u = k\mbox{ in }S \}
$$
is a $C^{\infty}$-submanifold of $C^{2,\a}(\o S).$
\end{proposition}

\begin{proposition}
\label{multiplier}
Let $\a\in (0, 1)$ and let $k > 0$ on $\o S$. If $v\in C_k^{2,\a}(\o S)$ is
stationary for $\W_k$, then there exists a unique Lagrange multiplier $\la\in (C^{\infty}\cap L^2)(S)$
such that
\begin{equation}\label{EL} 
\div\div \left( \chi_S \left( \D^2 v + \la\cof\D^2 v \right) \right) = 0\mbox{ in }\DD'(\R^2).
\end{equation}
In particular,
$$
\Delta^2 v + \cof\D^2 v : \D^2\la = 0 \mbox{ in the classical sense on }S.
$$
\end{proposition}

\section{The elliptic case}

\subsection{Functional analysis background}

In this section, $X$ and $Y$ denote real Banach spaces.
Recall that a closed subspace $E$ of $X$ is said to split $X$ if $E$
has a closed complement, i.e., there exists a closed subspace $F$ of $X$ such that
$X = E\oplus F$. For a linear operator $G : X\to Y$ we denote by $N(G)$ its kernel and by $R(G)$ its range.
The proof of the following result is straightforward.

\begin{lemma}\label{basiclagrange} 
Let $G : X\to Y$ and $F : X\to\R$ be bounded linear operators, and assume that the range of $G$ is closed.
Then 
$$
Fh = 0\mbox{ for all $h\in X$ with }Gh = 0
$$
if and only if there exists $\Lambda\in Y'$ such that 
$F = \Lambda\circ G$. If, moreover, $R(G) = Y$ then $\Lambda$ is unique.
\end{lemma}

Let $M\subset X$.
A vector $h\in X$ is called a tangent vector to $M$ at $v\in M$ provided there exists
a map $u$ from a neighbourhood of zero in $\R$ into $M$ such that $u(0) = v$,
and such that the derivative $u'(0)$ at $0$ exists and equals $h$. We denote
the set of all tangent vectors $h\in X$ at $v$ by $T_v M$.
We recall the following basic result.
\begin{lemma}\label{lju} 
Let $G : X\to Y$ be continuously Fr\'echet differentiable on $X$ and set
\begin{equation}
\label{lju-1} 
 M = \{ u\in X : G(u) = 0\}.
\end{equation}
Assume that, for all $v\in M$, the derivative $G'(v) : X\to Y$ is surjective
and the kernel $N\left( G'(v) \right)$ splits $X$.
Then $M$ is a $C^1$-manifold.
\\
More precisely, for all $v\in M$ we have $T_v M = N\left( G'(v) \right)$ and 
there exists a continuously Fr\'echet differentiable homeomorphism 
$\p$ from a neighbourhood of zero in $T_vM$ onto an
open neighbourhood of $v$ in $M$ that satisfies
$$
\p(h) = v + h + o\left( \|h\|_X \right)\mbox{ as }h\to 0\mbox{ in }T_v M.
$$
If $G$ is $C^m$ on $X$ then $M$ is a $C^m$-manifold.
\end{lemma}
\begin{proof}
This result is classical, cf. \cite{ljusternik}.
For the reader's convenience we recall the proof of the existence of $\p$.
\\
Let $E\subset X$ be a closed complement of $N\left( G'(v) \right)$ in $X$.
Define $H : N\left( G'(v) \right)\times E \to Y$ by setting $H(h,z) = G(v + h + z)$.
Since the partial derivative $D_2H(0, 0)$ is just the restriction of $G'(v)$ to $E$, 
the hypotheses show that we can apply the 
implicit function theorem to $H$. This yields a $C^1$-map $\p$
as in the statement, because $D_1H(0, 0) = 0$.
\end{proof}

{\bf Remark.} In the context of surfaces, the existence of $\p$ as in the conclusion of Lemma \ref{lju}
amounts to the so-called continuation of infinitesimal bendings. We refer to \cite{klimentov} for the
elliptic case (cf. also \cite{Nirenberg-Weyl} and \cite{LeMoPa}),
and to \cite{H-cpde} for the intrinsically flat case.

\subsection{Linear elliptic operators}

In this section we recall some basic functional analytic properties of 
linear elliptic operators of the form
$$
L u := A :\D^2 u + B\cdot\D u + Cu
$$
on a bounded $C^{2,\a}$ domain $\O\subset \R^n$ for some $\a\in (0, 1)$.
Here $A\in C^0(\o\O, \R^{n\times n}_{\sym})$, $B\in L^{\infty}(\O, \R^n)$ and $C\in L^{\infty}(\O)$.
We assume $A$ to be strictly elliptic, i.e.,
there exists $c > 0$ such that
$$
A(x):(\xi\otimes\xi) \geq c|\xi|^2\mbox{ for all }x\in S,\ \xi\in\R^n.
$$
For simplicity, we only consider the case when $C\leq 0$.

\begin{lemma}\label{splitwp}
Let $p\in (1, \infty)$, let $A\in C^0(\o\O)$, $B, C\in L^{\infty}(\O)$, assume that $C\leq 0$ 
and define $G : W^{2,p}(\O)\to L^p(\O)$ by
$
Gu = Lu.
$
Then $G$ is surjective and $N(G)$ splits $W^{2,p}(\O)$.
\end{lemma}
\begin{proof}
Under the above hypotheses on the coefficients, and for any $f\in L^p(\O)$,
the Dirichlet problem
\begin{equation}
\begin{split}\label{lemma4-1} 
Lu &= f
\\
u &\in W^{1,p}_0(\O) 
\end{split}
\end{equation}
has a unique solution $u\in W^{2,p}(\O)$, cf. \cite[Theorem 9.15]{GT}.
Hence $G$ is surjective.
\\
For each $f\in L^p(\O)$ denote by $Tf$ the solution $u$ of \eqref{lemma4-1}.
Then 
$T : L^p(\O) \to W^{2,p}(\O)$ is bounded. And obviously it is a right inverse of $G$.
Since $G$ is surjective and admits a bounded right inverse, 
we conclude that $N(G)$ splits $W^{2,p}(\O)$.
\end{proof}

Similarly, this time using Schauder theory, one proves the following lemma:
\begin{lemma}\label{split}
Assume $A$, $B$, $C\in C^{0, \a}(\o\O)$, assume that $C\leq 0$ and define 
$G : C^{2,\a}(\o\O)\to C^{0, \a}(\o\O)$ by $Gu = Lu$. Then $G$ is surjective and $N(G)$ splits $C^{2,\a}(\o\O)$.
\end{lemma}

\begin{lemma}\label{dense}
Assume $A$, $B$, $C\in C^{0, \a}(\o\O)$, assume that $C\leq 0$ and define 
$G : W^{2,2}(\O)\to L^2(\O)$ by $Gu = Lu$. Then $C^{2,\a}(\o\O)\cap N(G)$ is
strongly $W^{2,2}$-dense in $N(G)$.
\end{lemma}
\begin{proof}
Let $u\in N(G)$ and let $u_n\in C^{2,\a}(\o\O)$ be such that $u_n\to u$ in $W^{2,2}(\O)$.
Then by continuity
\begin{equation}
\label{dense-1}
Lu_n \to 0 \mbox{ in }L^2(\O).
\end{equation}
As in the proof of Lemma \ref{splitwp}, there exists a unique solution $\rho_n\in (W^{2,2}\cap W^{1,2}_0)(\O)$ of 
$$
L\rho_n = -Lu_n \mbox{ in }\O.
$$
Moreover, $\rho_n\in C^{2,\a}(\o\O)$ because $Lu_n\in C^{0,\a}(\o\O)$, and
$\rho_n\to 0$ in $W^{2,2}(\O)$ by \eqref{dense-1}. Thus $u_n + \rho_n\in N(G)\cap C^{2,\a}(\o\O)$
and $u_n + \rho_n \to u$ in $W^{2,2}(\O)$.
\end{proof}

\subsection{Proofs of the main results}

\begin{proof}[ of Proposition \ref{manif}]
Let $m\in\NN$ and set $X = C^{2,\a}(\o S)$ and $Y = C^{0, \a}(\o S)$. 
We must show that the map $G : X\to Y$ defined by
$$
G(u) = \det\D^2 u - k
$$
satisfies the hypotheses of Lemma \ref{lju}
with $M$ given by \eqref{lju-1}.
But of course $G$ is in $C^m$,
because it is quadratic. More precisely, for all $h\in X$ we have
$$
G(v + h) = G(v) + \cof\D^2 v : \D^2 h + \det\D^2 h.
$$
Since
$$
\|\det\D^2 h\|_{C^{0, \a}} \leq \|\D^2 h\|^2_{C^{0, \a}} \leq \|h\|^2_{C^{2, \a}},
$$
we have
$$
G'(v)h = \cof\D^2 v : \D^2 h \mbox{ for all }h\in X.
$$
Next we claim that $G'(v) : X\to Y$ is surjective for all $v\in M$. But since $\det\D^2 v = k$,
by the assumptions on $k$
we see that $v$ is either strictly convex or concave. Hence $\cof\D^2 v$ is strictly elliptic.
So Lemma \ref{split} shows that $G'(v)$ is surjective, and that $N\left( G'(v) \right)$ splits $X$.
\end{proof}

\begin{lemma}\label{lestat} 
Assume that $k > 0$ on $\o S$. If $v\in C^{2,\a}_k(\o S)$ is stationary for $\W_k$, then $v$ is formally stationary for $\W_k$.
\end{lemma}
\begin{proof}
Define $G : W^{2,2}(S) \to L^2(S)$ by $Gu = \cof\D^2 v : \D^2 u$ and $\t G : C^{2,\a}(\o S)\to C^{0, \a}(\o S)$ by $\t G u = \cof\D^2 v : \D^2 u$,
and define $F : W^{2,2}(S)\to \R$ by $F(v) = \int_S |\D^2 v|^2$.
\\
Since $F$ is continuously Fr\'echet differentiable, the fact that $v\in C_k^{2,\a}(\o S)$
is stationary for $\W_k$, combined with Proposition \ref{manif}
(in particular with the existence of $\p$ as in the conclusion of Lemma \ref{lju}), 
implies that
\begin{equation}
\label{lestat-1} 
F'(v)h = 0 \mbox{ for all }h\in N(\t G).
\end{equation}
Now $N(\t G) = N(G)\cap C^{2,\a}(\o S)$. Lemma \ref{dense}
implies that $N(\t G)$ is strongly $W^{2,2}$-dense in $N(G)$. Thus by continuity of $F'(v)$,
formula \eqref{lestat-1} is in fact equivalent to
$$
F'(v)h = 0 \mbox{ for all }h\in N(G).
$$
And this means that $v$ is formally stationary for $\W_k$.
\end{proof}

For formal stationary points we use the basic Lagrange multiplier rule to prove the following lemma:
\begin{lemma}\label{lemma3} 
Assume that $k > 0$ on $\o S$. If $v\in C^2_k(\o S)$ is formally stationary for $\W_k$, then there
exists a Lagrange multiplier $\la\in L^2(S)$ such that \eqref{EL} holds.
\end{lemma}
\begin{proof}
Define $F : W^{2,2}(S)\to\R$ by $Fh = \int_S \D^2 v : \D^2 h$ and
$G : W^{2,2}(S)\to L^2(S)$ by $Gh = \cof\D^2 v : \D^2 h$. 
By hypothesis we know $Fh = 0$ for all $h\in N(G)$. Since $\cof\D^2 v\in C^0(\o S)$,
we can apply Lemma \ref{splitwp} to see that $G$ is surjective.
Now Lemma \ref{basiclagrange} implies that there exists a unique $\Lambda$ in the dual of $L^2(S)$
such that $F = \Lambda\circ G$, i.e., there exists $\la\in L^2(S)$ such that
$$
\int_S \D^2 v : \D^2 h = \int_S \la \cof\D^2 v : \D^2 h \mbox{ for all }h\in W^{2,2}(S).
$$
\end{proof}

\begin{proof}[ of Proposition \ref{multiplier}]
Combine Lemma \ref{lemma3} with
Lemma \ref{lestat}, and observe that \eqref{EL} implies that
$$
\div\div \left( \la \cof\D^2 v : \D^2 h \right) = \Delta^2 v \mbox{ in }\DD'(S),
$$
which by standard elliptic regularity proves
that $\la\in C^{\infty}(S)$ because $v\in C^{\infty}(S)$.
\end{proof}

\section{The case of constant $k$}

For $A\in\R^{2\times 2}_{\sym}$ we have $|A|^2 = (\Tr A)^2 - 2\det A$
and
$
|A|^2 = 2|A^{\circ}|^2 + 2\det A,
$
where $A^{\circ} = A - \frac{1}{2}(\Tr A) I$ denotes the trace-free part of $A$.
Thus
\begin{equation*}
|\D^2 v|^2 = (\Delta v)^2 - 2k
\end{equation*} 
and
\begin{equation*}
|\D^2 v|^2 = 2|\D^2 v - \frac{1}{2}\Delta v\ I|^2 + 2k.
\end{equation*} 
And so
$$
\W_k(v) = \int_S \left( (\Delta v)^2 - 2k \right) = 
2\int_S \left(|\D^2 v - \frac{1}{2}\Delta v\ I|^2 + k \right).
$$
So an absolute minimum of $\W_k$ is attained
if $\D^2 v - \frac{1}{2}\Delta v\ I$
vanishes identically, i.e.,
$$
\D^2 v = 
\begin{cases}
\sqrt{k}I &\mbox{ if }k \geq 0
\\
\sqrt{|k|}\diag\left(1, -1 \right) &\mbox{ if }k < 0.
\end{cases}
$$
If $k$ is constant then this is the case for
$$
v(x) =
\begin{cases}
\frac{\sqrt{k}}{2} |x|^2 &\mbox{ if }k\geq 0
\\
\frac{\sqrt{|k|}}{2} (x_1^2 - x_2^2) &\mbox{ if }k < 0.
\end{cases}
$$

{\bf Remarks.}
\begin{enumerate}[(i)]
\item The above computations are standard in the context of surfaces, cf. 
\cite{KS-annals}. With obvious changes, these arguments also apply to the case
of isometric immersions when the Gauss curvature of the reference metric is constant,
cf. \cite{H-general}.
\item When the energy density $Q_2$ is not isotropic, then one can still argue similarly,
and one obtains solutions with two unequal constant principal curvatures.
\item A nontrivial problem for the case $k = 0$ results if one
imposes boundary conditions or includes force terms.
This situation is covered by the results in \cite{H-cpam, H-cpde}.
\\
Indeed, the problem addressed there was to study minimisers of the Willmore functional 
$$
u\mapsto\int_S |A|^2
$$
among all $W^{2,2}$ isometric immersions $u$ of $(S, \delta)$ into $\R^3$, where $\delta$ denotes the standard flat metric in $\R^2$
and $A$ denotes the second fundamental form of $u$.
But by the Gauss-Codazzi-Mainardi equations
$A$ is a possible second fundamental form for such an isometric immersion $u$ if and only if
$$
A\in\{\D^2 v : v\in W^{2,2}(S)\mbox{ with }\det\D^2 v = 0\mbox{ a.e. in }S\}.
$$
Related to this, if $u\in W^{2,2}(S, \R^3)$,
then $u$ is an isometric immersion if and only if the function $v = u\cdot e$ 
satisfies the Darboux equation $\det\D^2 v = 0$ for any constant $e\in\SS^2$.
\end{enumerate}

\vspace{1cm}

\def\cprime{$'$}


\begin{thebibliography}{10}

\bibitem{Ciarlet2}
P.~G. Ciarlet.
\newblock {\em Mathematical elasticity. {V}ol. {II}}, volume~27 of {\em Studies
  in Mathematics and its Applications}.
\newblock North-Holland Publishing Co., Amsterdam, 1997.
\newblock Theory of plates.

\bibitem{EfratiSharonKupferman}
E.~Efrati, E.~Sharon, and R.~Kupferman.
\newblock Elastic theory of unconstrained non-euclidean plates.
\newblock {\em J. Mech. Phys. Solids}, 57:762--775, 2009.

\bibitem{fjm1}
G.~Friesecke, R.~D. James, and S.~M{\"u}ller.
\newblock A theorem on geometric rigidity and the derivation of nonlinear plate
  theory from three-dimensional elasticity.
\newblock {\em Comm. Pure Appl. Math.}, 55(11):1461--1506, 2002.

\bibitem{fjm2}
G.~Friesecke, R.~D. James, and S.~M{\"u}ller.
\newblock A hierarchy of plate models derived from nonlinear elasticity by
  gamma-convergence.
\newblock {\em Arch. Ration. Mech. Anal.}, 180(2):183--236, 2006.

\bibitem{GT}
D.~Gilbarg and N.~S. Trudinger.
\newblock {\em Elliptic partial differential equations of second order}.
\newblock Classics in Mathematics. Springer-Verlag, Berlin, 2001.
\newblock Reprint of the 1998 edition.

\bibitem{GuvenMuller}
J.~Guven and M.~M. M{\"u}ller.
\newblock How paper folds: bending with local constraints.
\newblock {\em J. Phys. A}, 41(5):055203, 15, 2008.

\bibitem{H-cpam}
P.~Hornung.
\newblock Euler-{L}agrange equation and regularity for flat minimizers of the
  {W}illmore functional.
\newblock {\em Comm. Pure Appl. Math.}, 64(3):367--441, 2011.

\bibitem{H-convex}
P.~Hornung.
\newblock Constrained {W}illmore equation for disks with positive {G}auss
  curvature.
\newblock {\em MIS MPG Preprint}, 2012.

\bibitem{H-general}
P.~Hornung.
\newblock The {W}illmore functional on isometric immersions.
\newblock {\em MIS MPG Preprint}, 2012.

\bibitem{H-cpde}
P.~Hornung.
\newblock Continuation of infinitesimal bendings on developable surfaces and
  equilibrium equations for nonlinear bending theory of plates.
\newblock {\em Comm. PDE}, to appear.

\bibitem{klimentov}
S.~B. Klimentov.
\newblock Extension of higher-order infinitesimal bendings of a simply
  connected surface of positive curvature.
\newblock {\em Mat. Zametki}, 36(3):393--403, 1984.

\bibitem{KS-annals}
E.~Kuwert and R.~Sch{\"a}tzle.
\newblock Removability of point singularities of {W}illmore surfaces.
\newblock {\em Ann. of Math. (2)}, 160(1):315--357, 2004.

\bibitem{LeMaPa}
M.~Lewicka, L.~Mahadevan, and M.~R. Pakzad.
\newblock Models for elastic shells with incompatible strains.
\newblock {\em Preprint}, 2012.

\bibitem{LeMoPa}
M.~Lewicka, M.~G. Mora, and M.~R. Pakzad.
\newblock The matching property of infinitesimal isometries on elliptic
  surfaces and elasticity of thin shells.
\newblock {\em Arch. Ration. Mech. Anal.}, 200(3):1023--1050, 2011.

\bibitem{ljusternik}
L.~Ljusternik.
\newblock On constrained extrema of functionals.
\newblock {\em Mat. Sb.}, 41(11):390--401, 1934.

\bibitem{Nirenberg-Weyl}
L.~Nirenberg.
\newblock The {W}eyl and {M}inkowski problems in differential geometry in the
  large.
\newblock {\em Comm. Pure Appl. Math.}, 6:337--394, 1953.

\end{thebibliography}
\end{document}